\documentclass[10pt,a4paper,utf8]{article}
\pdfoutput=1

\usepackage[a4paper,inner=100pt,outer=100pt,textheight=630pt,centering]{geometry}

\usepackage{enumerate}
\usepackage[usenames,dvipsnames]{xcolor}
\usepackage[utf8]{inputenc}
\usepackage{lmodern}
\usepackage[T1]{fontenc}

\usepackage{color}
\usepackage{url}
\usepackage{amssymb}
\usepackage{amsmath}
\usepackage{amsthm}
\usepackage{cases}
\usepackage[innertopmargin=-5,skipbelow=-5]{mdframed}
\usepackage{stmaryrd}
\usepackage{mathpartir}
\usepackage{tikz}
\usepackage{tikz-cd}
\usetikzlibrary{arrows,matrix,decorations.pathmorphing,
  decorations.markings, calc, backgrounds}

\theoremstyle{plain}
\newtheorem{theorem}{Theorem}
\newtheorem*{theorem*}{Theorem}
\newtheorem{lemma}[theorem]{Lemma}
\newtheorem*{lemma*}{Lemma}
\newtheorem{corollary}[theorem]{Corollary}
\newtheorem*{corollary*}{Corollary}
\newtheorem{proposition}[theorem]{Proposition}
\newtheorem*{proposition*}{Proposition}

\theoremstyle{definition}
\newtheorem{definition}[theorem]{Definition}
\newtheorem*{definition*}{Definition}

\newtheorem*{example*}{Example}

\theoremstyle{remark}

\newtheorem*{remark*}{Remark}

\newcommand{\cubicaltt}{{\sc cubicaltt}}

\newcommand{\po}{\mathsf{po}}

\newcommand{\elim}{\mathsf{elim}}

\newcommand{\der}{\vdash}
\newcommand{\co}{\colon}

\newcommand{\all}{\forall}

\newcommand{\Imp}{\Rightarrow}

\def\phi{\varphi}

\newcommand{\is}{\vec i}

\newcommand{\rs}{\vec r}

\newcommand{\us}{\vec u}
\newcommand{\vs}{\vec v}

\newcommand{\xs}{\vec x}

\newcommand{\zs}{\vec z}

\newcommand{\thetas}{\vec \theta}

\newcommand{\As}{\vec A}

\newcommand{\Ps}{\vec P}

\DeclareMathOperator{\dM}{\mathsf{dM}}
\newcommand{\II}{\mathbb{I}}
\newcommand{\FF}{\mathbb{F}}

\newcommand{\su}[2]{#1/#2}
\newcommand{\subst}[2]{(\su #1 #2)}

\newcommand{\pp}{\mathsf{p}}       
\newcommand{\qq}{\mathsf{q}}       

\newcommand{\N}{\mathsf{N}}

\newcommand{\unit}{\mathsf{t\!t}}         

\newcommand{\Path}{\mathsf{Path}}
\newcommand{\nabs}[1]{\langle #1 \rangle}

\newcommand{\tabs}[1]{[ #1 ]}

\newcommand{\comp}{\mathsf{comp}}

\newcommand{\hcomp}{\mathsf{hcomp}}
\newcommand{\hComp}{\mathsf{hfill}}

\newcommand{\transport}{\mathsf{transport}}

\def\transp{\mathsf{trans}}
\newcommand{\ctransp}{\mathsf{ctrans}}

\newcommand{\Transp}{\mathsf{transFill}}
\newcommand{\cTransp}{\mathsf{ctransFill}}

\newcommand{\Glue}{\mathsf{Glue}} 

\newcommand{\UU}{\mathsf{U}}

\newcommand{\squeeze}{\mathsf{squeeze}}

\newcommand{\CC}{\mathcal{C}}

\DeclareMathOperator{\Fib}{\mathsf{Fibrant}} 

\DeclareMathOperator{\yoneda}{\mathbf{y}}
\DeclareMathOperator{\app}{app}

\DeclareMathOperator{\scomp}{\mathsf{c}} 
\DeclareMathOperator{\sfill}{\mathsf{f}} 
\DeclareMathOperator{\CompStr}{\mathsf{Comp}} 
\DeclareMathOperator{\FillStr}{\mathsf{Fill}} 

\newcommand{\strans}{t}      
\newcommand{\shcomp}{h}    

\definecolor{dkblue}{rgb}{0,0.1,0.5}
\definecolor{lightblue}{rgb}{0,0.5,0.5}
\definecolor{dkgreen}{rgb}{0,0.4,0}
\definecolor{dk2green}{rgb}{0.4,0,0}
\definecolor{dkviolet}{rgb}{0.6,0,0.8}

\newcommand{\Sp}{{\mathbb S}}
\newcommand{\base}{{\sf base}}
\newcommand{\LOOP}{{\sf loop}}
\newcommand{\Spelim}[2]{\Sp^#1\textsf{-elim}_{#2}\,}

\DeclareMathOperator{\susp}{\mathsf{Susp}}
\newcommand{\north}{{\sf N}}
\newcommand{\south}{{\sf S}}
\newcommand{\merid}{{\sf merid}}

\newcommand{\TT}{{\mathbb T}}
\newcommand{\tb}{{\sf b}}
\newcommand{\tp}{{\sf tp}}
\newcommand{\tq}{{\sf tq}}
\newcommand{\tsurf}{{\sf surf}}

\newcommand{\TTF}{\TT_{\sf F}}
\newcommand{\tfb}{{\sf \tb_{\sf F}}}
\newcommand{\tfp}{{\sf \tp_{\sf F}}}
\newcommand{\tfq}{{\sf \tq_{\sf F}}}
\newcommand{\tfsurf}{\tsurf_{\sf F}}

\newcommand{\inh}[1]{\lVert#1\rVert}
\DeclareMathOperator{\inc}{\mathsf{inc}}
\newcommand{\squash}{\mathsf{sq}}

\newcommand{\pushA}{A}
\newcommand{\pushB}{B}
\newcommand{\pushC}{C}
\newcommand{\pushf}{u}
\newcommand{\pushg}{v}
\newcommand{\Push}{\pushA \sqcup_\pushC \pushB}
\DeclareMathOperator{\inl}{\mathsf{inl}}
\DeclareMathOperator{\inr}{\mathsf{inr}}
\newcommand{\push}{\mathsf{push}} 

\newcommand{\gD}{\mathsf{D}}     
\newcommand{\gc}{\mathsf{c}}     
\newcommand{\gf}{\mathsf{f}}     

\begin{document}
\title{On Higher Inductive Types in Cubical Type Theory}

\author{Thierry Coquand, Simon Huber, and Anders Mörtberg}

\date{}

\maketitle

\begin{abstract}
  Cubical type theory provides a constructive justification to certain
  aspects of homotopy type theory such as Voevodsky's univalence
  axiom. This makes many extensionality principles, like function and
  propositional extensionality, directly provable in the theory. This
  paper describes a constructive semantics, expressed in a presheaf
  topos with suitable structure inspired by cubical sets, of some
  higher inductive types. It also extends cubical type theory by a
  syntax for the higher inductive types of spheres, torus,
  suspensions, truncations, and pushouts. All of these types are
  justified by the semantics and have judgmental computation rules for
  all constructors, including the higher dimensional ones, and the
  universes are closed under these type formers.
\end{abstract}

\tableofcontents

\section{Introduction}
\label{sec:introduction}

Homotopy type theory~\cite{HottBook13} provides a new and promising
approach to equality in type theory where types are thought of as
abstract spaces and equality as paths in these
spaces~\cite{AwodeyWarren09}. Iterated equality proofs then correspond
to \emph{homotopies} between paths. This intuition is motivated by
homotopy theoretic models, in particular by the Kan simplicial set
model~\cite{KapulkinLumsdaine12} due to Voevodsky. This allows one to
find new principles in type theory inspired by homotopy theory. Prime
examples of this are Voevodsky's \emph{univalence
  axiom}~\cite{Voevodsky10}, which generalizes the principle of
propositional extensionality to dependent type theory, and the
stratification of types by the complexity of their equality (i.e., by
their homotopy level or ``h-level''~\cite{Voevodsky15}).

In the homotopical interpretation of type theory inductive types are
represented as discrete spaces with only points in them. Higher
inductive types are a natural generalization where types may also be
generated by paths (potentially higher dimensional). This notion of
types, combined with universes and the univalence axiom, is an
important extension of dependent type theory, which allows for an
elegant and original synthetic development of algebraic topology,
using in a key way type-theoretic ideas (such as the encode-decode
method~\cite{HottBook13}). Impressive examples of this development
are, among others, the definition of the Hopf fibration, the
Freudenthal suspension theorem and the Blakers-Massey
theorem~\cite{Brunerie16,BlakersMassey}. However, and somewhat
surprisingly, despite several efforts (e.g., \cite{LumsdaineShulman}),
the {\em consistency} of such an extension, which would justify these
impressive developments, has not yet been established. The simplicial
set model~\cite{KapulkinLumsdaine12} provides (in a classical
framework) a model for the univalence axiom, but it only provides a
model for some very particular higher inductive types (such as the
spheres, and the propositional truncation via an impredicative
encoding~\cite{Voevodsky15}), and, as explained in
\cite{LumsdaineShulman}, it is not clear how to extend this model to a
model of parametrized higher inductive types like the suspension or
pushouts (expressed as operations on a given universe).

\paragraph{Contributions}

The first contribution of the present paper is to provide such a
semantics, starting in an essential way not from the simplicial set
model, but from a cubical set
model~\cite{CohenCoquandHuberMortberg15,OrtonPitts16}. This semantics
is furthermore carried out in a constructive meta-theory. Our second
contribution is to extend cubical type theory with a syntax for higher
inductive types, exemplified by: spheres, the torus, suspensions,
truncations, and pushouts. These types illustrate many of the
difficulties in giving a computational justification for a general
class of higher inductive types, in particular: the spheres and torus
have higher dimensional constructors, furthermore one version of the
torus refers to ``fibrancy'' structure in its endpoints, the
suspension has a parameter type, the truncations are recursive, and
the pushouts have function applications in the endpoints of the path
constructor. We show how to overcome all of these difficulties in a
uniform way which suggests an approach to the problem of defining a
schema for higher inductive types in cubical type theory.

Furthermore, all of the higher inductive types we consider have the
following good properties justified by our semantics:
\begin{enumerate}
\item judgmental computation rules for all constructors,
\item strict stability under substitution, and
\item closure under universe levels (the higher inductive types live
  in the same universe as their parameters).
\end{enumerate}

We have also implemented a variation of the system presented in this
paper and performed multiple experiments with
it.\footnote{See:~\url{https://github.com/mortberg/cubicaltt/tree/hcomptrans}}

\paragraph{Outline}

The paper begins by describing the semantics, expressed in a presheaf
topos with suitable structure, of the circle
(Section~\ref{sec:semcircle}), suspension (Section~\ref{sec:semsusp}),
and pushouts (Section~\ref{sec:sempush}). The next section starts with
a short background on cubical type theory
(Section~\ref{sec:cubicaltt}) followed by the extension to the theory
with: circle and spheres (Section~\ref{subsec:spheres}), the torus
(Section~\ref{subsec:torus}), suspensions
(Section~\ref{subsec:suspension}), propositional truncation
(Section~\ref{subsec:truncations}), and pushouts
(Section~\ref{subsec:pushouts}). The paper ends with conclusions and
discussions on future and related work
(Section~\ref{sec:conclusions}).

\section{Semantics of higher inductive types}
\label{sec:semantics}

As shown in \cite{OrtonPitts16,LicataEtAl18,AngiuliEtAl17}, the
presentation of the semantics of cubical type theory can be both
simplified and clarified by using the language of extensional type
theory (with universes). This language can be given meaning in any
presheaf topos, so long as we assume that the ambient set theory has a
hierarchy of Grothendieck universes.  In particular, we are going to
show that the justification of higher inductive types can be done
internally, using the existence of suitable initial algebras as the
only extra assumption.  We then justify the existence of these initial
algebras for our presheaf topos externally. The key idea will be a
decomposition of the notion of composition structure
\cite{OrtonPitts16,LicataEtAl18} into a {\em transport} and a {\em
  homogeneous composition} operation.\footnote{As explained in
  \cite{AngiuliEtAl17} this decomposition was first introduced in an
  early version of \cite{CohenCoquandHuberMortberg15}, precisely to
  address the problem of the semantics of propositional truncation and
  this decomposition is also present
  in~\cite{CHTTPOPL,AngiuliHouHarper17,CavalloHarper18}.} This
decomposition can be described internally.

We will work here in the presheaf topos over the Lawvere theory of
De~Morgan algebras \cite{CohenCoquandHuberMortberg15,LicataEtAl18}
(but, following \cite{OrtonPitts16}, our results are valid in a more
general setting).  The presentation we use in
\cite{CohenCoquandHuberMortberg15} of this category is the following:
we fix a countable set of names/symbols and the objects of the
category $I,J,\dots$ are finite sets of symbols. A map $J\rightarrow
I$ is then a set-theoretic map from $I$ to the free De~Morgan algebra
$\dM(J)$ on $J$.  The corresponding presheaf model has then a generic
De~Morgan algebra $\II$, taking $\II(J)$ to be $\dM(J)$.  (To have
such a structure on $\II$ is not strictly
necessary~\cite{OrtonPitts16}, but it simplifies the presentation.)

This type $\II$ is used as an abstract representation of the unit
interval, so that a path in a type $A$ is represented by an element of
the exponential $A^{\II}$.  The extra data needed to define a cubical
set model is a notion of \emph{cofibration}, which specifies the
shape of filling problems that can be solved in a dependent type.  We
represent this by a type of \emph{cofibrant} propositions $\FF$
(denoted by $\mathtt{Cof}$ in~\cite{OrtonPitts16}). In
\cite{CohenCoquandHuberMortberg15}, this is represented by the {\em
  face lattice} (see Section~\ref{sec:cubicaltt}), but other choices
are possible.  (Classically, this type $\FF$ is a subtype of the
subobject classifier of the presheaf topos, but, as stressed in
\cite{LicataEtAl18}, we can avoid mentioning the impredicative type of
propositions altogether, and work in a predicative meta-theory.)  We
write $[\phi]$ for the type associated to the proposition $\phi :
\FF$.  So $[\phi]$ is a sub-singleton, and any element of $[\phi]$ is
equal to a fixed element~$\unit$.

A \emph{partial element} of a type $T$ is given by an element $\phi$
in $\FF$ and a function $[\phi] \to T$.  We say that a total element
$v$ of $T$ extends such a partial element $\phi, u$ if we have $\phi
\Imp u\,\unit = v$, where $\Imp$ denotes implication between
propositions.

In this extensional type theory, we can think of a dependent type $A$
over a given type $\Gamma$ as a family of types $A\rho$ indexed by
elements $\rho$ of $\Gamma$.

We now recall the notions of composition and filling
structures~\cite{CohenCoquandHuberMortberg15,OrtonPitts16}.  Let $A$
be a dependent type over a type $\Gamma$.

\begin{definition}
  \label{def:comp-structure}
  A \emph{composition structure} $\scomp_A$ on $A$ is an operation
  taking as inputs $\gamma$ in $\Gamma^\II$, a proposition $\phi$ in
  $\FF$, a partial element $u$ in $[\phi] \to \Pi (i : \II)\, A \gamma
  (i)$, and an element $u_0$ in $A \gamma (0)$ such that $\phi \Imp
  u\,\unit\,0 = u_0$.  This operation produces an element $u_1 =
  \scomp_A\,\gamma\,\phi\,u\,u_0$ in $A \gamma (1)$ such that $\phi
  \Imp u\,\unit\,1 = u_1$.

  The type of all such operations is written $\CompStr(\Gamma,A)$
  (see~\cite[Definition 4.3]{OrtonPitts16} for an explicit internal
  definition).
\end{definition}

\begin{definition}
  \label{def:fill-structure}
  A \emph{filling structure} $\sfill_A$ on $A$ is an operation taking
  the same input as $\scomp_A$ above, but producing an element $v =
  \sfill_A\,\gamma\,\phi\,u\,u_0$ in $\Pi(i:\II)\,A\gamma(i)$ such
  that $v$ extends $u$, i.e., $\phi \Imp u\,\unit = v$, and $v~0 = u_0$.

  We write $\FillStr(\Gamma,A)$ for the type of filling structures on
  $A$.
\end{definition}

This notion of filling structure is an internal form of the homotopy
extension property, which was recognized very early (see, e.g.,
\cite{Eilenberg1939}) as a key for an abstract development of
algebraic topology.

As explained in~\cite{CohenCoquandHuberMortberg15,OrtonPitts16} we
have that $\CompStr(\Gamma,A)$ is a retract of $\FillStr(\Gamma,A)$.
\medskip

In the particular case where $\Gamma$ is the unit type, then $A$ is a
``global'' type, and $\CompStr (\Gamma,A)$ becomes the type $\Fib(A)$
expressing that $A$ is a {\em fibrant object}. Such a \emph{fibrancy
  structure} on $A$ consists of an operation $h_A$ taking as arguments
$u_0$ in $A$ and a partial element $\varphi,u$ of $A^{\II}$ such that
$\varphi\Rightarrow u~\unit~0 = u_0$, and produces an element
$u_1 = h_A~\varphi~u~u_0$ such that
$\varphi\Rightarrow u~\unit~1 = u_1$.

In general, if $A$ is a family of types over $\Gamma$, to give a
composition structure for each fiber, that is, an element in $\Pi
(\rho:\Gamma)\Fib(A\rho)$, is not enough to get a global composition
structure, that is, an element in $\CompStr(\Gamma,A)$ (see
\cite{OrtonPitts16} for an explicit counterexample).  An element in
$\Pi (\rho:\Gamma)\Fib(A\rho)$ is called a {\em homogeneous
  composition structure}.

\medskip

We now describe the notion of {\em transport} operation, which allows
to define a composition structure from a homogeneous composition
structure.
This decomposition of the composition operation into a transport and
homogeneous composition operation plays a crucial role for
interpreting higher inductive types depending on parameters (like
suspension, pushouts, or propositional truncation).

\begin{definition}
  \label{def:transp-structure}
  A \emph{transport structure} $\strans_A$ on $A$ is an operation
  taking as arguments a path $\gamma$ in $\Gamma^\II$, a proposition
  $\phi$ in $\FF$ such that $\phi \Imp \all (i~:~\II)\,\gamma (0) =
  \gamma (i)$, and an element $u_0$ in $A \gamma (0)$.  This operation
  produces an element $u_1 = \strans_A\,\gamma\,\phi\,u_0$ in $A
  \gamma (1)$ such that $\phi \Imp u_0 =~u_1$.
\end{definition}

The condition $\varphi\Rightarrow \forall (i:\II)\, \gamma(0) =
\gamma(i)$ expresses that the path $\gamma$ is {\em constant on}
$\varphi$.

Clearly we obtain a homogeneous composition structure from any
composition structure.  We also get:
\begin{lemma}
  If a family of types $A$ over $\Gamma$ has a composition structure
  $\scomp_A$, then it has a transport structure $\strans_A$.
\end{lemma}
\begin{proof}
  We can take $\strans_A\,\gamma\,\phi\,u_0 =
  \scomp_A\,\gamma\,\phi\,(\lambda (x : [\phi]) (i:\II)\,u_0)\,u_0$.
\end{proof}

\begin{lemma}\label{lemma:compfromtransp}
  If a family of types $A$ over $\Gamma$ has a homogeneous composition
  structure $\shcomp_A$ and a transport structure $\strans_A$, then it
  has a composition structure $\scomp_A$.
\end{lemma}
\begin{proof}
  We can define $\scomp_A\,\gamma\,\phi\,u\,u_0$ as
  \begin{equation*}
    \shcomp_A\,\gamma (1)\,\phi\, %
    (\lambda (x:[\phi]) (i : \II).\,
    \strans_A\,\gamma'(i)\,(i=1)\,(u\,x\,i))\,%
    (\strans_A\,\gamma\,0_\FF\,u_0)
  \end{equation*}
  where $\gamma' (i) = \lambda (j : \II) \, \gamma (i \lor j)$.
\end{proof}

We are now going to develop some universal algebra internally in the
presheaf model. The operations will involve the interval object $\II$
and the type $\FF$ of cofibrant propositions, and can be seen as a
generalization of the usual notion of operations in universal algebra.

\subsection{Semantics of the circle}
\label{sec:semcircle}

The circle, denoted $\Sp^1$, is represented as a higher inductive type
with a path $\LOOP$ in direction $i : \II$ connecting a point $\base$
to itself:
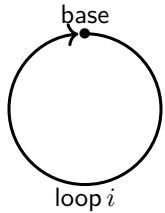
\begin{mathpar}
  \begin{tikzpicture}
    \draw[->,thick,draw=black,solid,line width=0.4mm]
    (0:0) node[above]{$\base$} arc[radius=1, start angle=90, end
    angle=-265];
    \node[] at (0,0) (b) {$\bullet$};
    \node[] at (0,-2.2) (b) {$\LOOP\,i$};
   \end{tikzpicture}
\end{mathpar}

If $A$ (resp. $B$) has a fibrancy structure $h_A$ (resp. $h_B$), then
a map $\alpha:A\rightarrow B$ is \emph{fibrancy preserving} if it satisfies
\[
  \alpha~(h_A~\varphi~u~u_0) = h_B~\varphi~(\lambda
  (x:[\varphi])(i:\II)~\alpha\,(u~x~i))~(\alpha~u_0).
\]

An \emph{$S^1$-algebra structure} on a type $A$ consists of a fibrancy
structure $h_A$ together with a base point $b_A$ and a loop $l_A$ in
$A^{\II}$ connecting $b_A$ to itself (i.e., $l_A~0 = l_A~1 =
b_A$). Given two $S^1$-algebras $A,h_A,b_A,l_A$ and $B,h_B,b_B,l_B$ a
function $\alpha:A\rightarrow B$ is a map of $S^1$-algebras if it is
fibrancy preserving and
satisfies $\alpha~b_A = b_B$ and $\alpha~(l_A~i) = l_B~i$.

We will show below using {\em external} reasoning:

\begin{proposition}\label{init1}
  There exists an initial $S^1$-algebra, denoted by
  $\Sp^1,\hcomp,\base,\LOOP$.
\end{proposition}

So $\Sp^1$ has a structure of an $S^1$-algebra and the fact that it is
initial means that, for any $S^1$-algebra $A,h_A,b_A,l_A$ there exists
a unique $S^1$-algebra map $\Sp^1\rightarrow A$.

By definition, the type $\Sp^1$ is fibrant since it has a fibrancy
structure $\hcomp$.  Furthermore, we can prove that initiality implies
the dependent elimination rule.\footnote{This is a direct
  generalization of the usual argument that a natural number object
  satisfies the dependent elimination rule.}

\begin{proposition}
  $\Sp^1$ satisfies the dependent elimination rule for the circle:
  given a family of types $P$ over $\Sp^1$ with a composition
  structure, and $a$ in $P~\base$ and $l~i$ in $P~(\LOOP~i)$ such that
  $l~0 = l~1 = a$ there exists a map $\elim:\Pi (x:\Sp^1)P~x$ such
  that $\elim~\base = a$ and $\elim~(\LOOP~i) = l~i$.
\end{proposition}
\begin{proof}
  We know by \cite{CohenCoquandHuberMortberg15,OrtonPitts16} that $A =
  \Sigma (x:\Sp^1)P~x$ has a composition structure. It has then a
  natural $\Sp^1$-algebra structure, taking $b_A = \base,a$ and $l_A~i
  = \LOOP~i,l~i$. This structure is such that the first projection
  $\pi_1:A\rightarrow \Sp^1$ is a map of $S^1$-algebras. We have a
  unique $S^1$-algebra map $\alpha:\Sp^1\rightarrow A$ and $\pi_1\circ
  \alpha$ is the identity on $\Sp^1$. We can then define $\elim~x =
  \pi_2~(\alpha~x)$ in $P~x$.
\end{proof}

\subsection{Semantics of the suspension operation}
\label{sec:semsusp}

The suspension $\susp{A}$ of a type $A$ has constructors $\north$ and
$\south$ (two poles) and a path between them for any element of
A. This enables us to give a direct definition of $\Sp^{n+1}$ as
$\susp^n\Sp^1$.  Compared to the circle, this higher inductive type
presents the extra complexity of having parameters and the
decomposition of the composition operation will be the key for
providing its semantics.

Given a type $X$, a \emph{$\susp{X}$-algebra structure} on a type $A$
consists of a fibrancy structure $h_A$ together with two points
$n_A,s_A$, and a family of paths $l_A$ in $X\rightarrow A^{\II}$
connecting $n_A$ to $s_A$ (i.e., $l_A~x~0 = n_A$ and $l_A~x~1 = s_A$
for all $x$ in $X$). Given two $\susp{X}$-algebras $A,h_A,n_A,s_A,l_A$
and $B,h_B,n_B,s_B,l_B$ a function $\alpha \co A \to B$ is a map of
$\susp{X}$-algebras if it is fibrancy preserving
and satisfies $\alpha~n_A =~n_B$, $\alpha~s_A = s_B$,
and $\alpha~(l_A~i) = l_B~i$.

As for the circle we can show using external reasoning:

\begin{proposition}\label{init2}
  There exists an initial $\susp{X}$-algebra, denoted by
  $\susp{X},\hcomp,\north,\south,\merid$.
\end{proposition}

By definition, the type $\susp{X}$ is fibrant since it has a fibrancy
structure $\hcomp$.  Using this filling structure, we prove as above:

\begin{proposition}
  $\susp{X}$ satisfies the dependent elimination rule for the
  suspension: given a family of type $P$ over $\susp{X}$ with a
  composition structure, and $n$ in $P~\north$ and $s$ in $P~\south$
  and $l~x~i$ in $P~(\merid~x~i)$ such that $l~x~0 = n$ and $l~x~1 =
  s$ there exists a map $\elim:\Pi (x:\susp{X})P~x$ such that
  $\elim~\north = n$ and $\elim~\south = s$ and $\elim~(\merid~x~i) =
  l~x~i$.
\end{proposition}

The operation $\susp{X}$ is functorial in $X$. Given a map
$u:X\rightarrow Y$ we get a $\susp{X}$-structure on $\susp{Y}$ by
taking $l_{\susp{Y}}~x~i = \merid_Y~(u~x)~i$ and hence a map $\susp(u)
: \susp{X}\rightarrow \susp{Y}$.

Let now $A$ be a dependent family of types over a given type $\Gamma$,
so that $A\rho$ is a type for any $\rho$ in $\Gamma$. We define a new
family of types $\susp{A}$ over $\Gamma$ by taking $(\susp{A})\rho =
\susp (A\rho)$. By construction, this new family {\em always} has a
homogeneous composition structure (without any hypothesis on $A$).

\begin{proposition}
  If $A$ has a transport structure $\strans_A$, then $\susp{A}$ has a
  transport structure, and hence (since it has a homogeneous
  composition structure) also a composition structure by
  Lemma~\ref{lemma:compfromtransp}.
\end{proposition}
\begin{proof}
  Given $\gamma$ in $\Gamma^{\II}$ and $\varphi$ such that $\gamma$ is
  constant on $\varphi$ (i.e., $\varphi\Rightarrow \forall
  (i:\II)~\gamma(0) = \gamma(i)$), we have a map
  $t_A~\gamma~\varphi:A\gamma(0)\rightarrow A\gamma(1)$ which is the
  identity on $\varphi$ and hence the map $\susp{(t_A~\gamma~\varphi)}$ is a transport map
  $\susp{(A\gamma(0))}\to \susp{(A\gamma(1))}$ which is the identity on
  $\varphi$.
\end{proof}

This example motivates the decomposition of the composition operation
into a transport and homogeneous composition operations. In a context,
we could only build an initial algebra for the {\em homogeneous}
composition operation (by doing it pointwise) and it does not seem
possible to do it for the composition operation directly.  The problem
does not appear for a type like the circle which has no parameters,
for which homogeneous and general compositions coincide.  For the
suspension however, we have to argue further that we also get a transport
operation.  (This problem seems connected to the problem of size
blow-up for parametrized higher inductive types due to fibrant
replacement in the simplicial set model discussed
in~\cite{LumsdaineShulman}.)

\medskip

The same argument applies to the propositional truncation
$\inh{X}$ of a type $X$. We would then instead consider the following
notion of algebra: a type $A$ with a fibrancy structure, a map
$i_A:X\rightarrow A$ and a map $sq_A:A\rightarrow A\rightarrow
A^{\II}$ such that $sq~a_0~a_1$ is a path connecting $a_0$ to $a_1$.

\subsection{Pushouts}
\label{sec:sempush}

Many examples of higher inductive types can be encoded as (homotopy) pushouts of
spans of other types. In particular (homotopy) coequalizers, which
together with coproducts (which are encoded using $\Sigma$-types), can
be used to compute general colimits of diagrams of types. This has
been used to encode many known higher inductive types, including
recursive ones like propositional~\cite{Doorn16,Kraus16} and higher
truncations~\cite{Rijke17}.

The semantics of pushouts involves the same problem with parameters as
in the previous example, but the definition of the transport function
is more complex and we will need to introduce some auxiliary
operations definable from transport.

A {\em span} $D = (C,A,B,u,v)$ consists of two maps $u \co C \to A$
and $v \co C \to B$. Given such a span, we define a $D$-algebra to be
a type $X$ with a fibrancy structure $h_X$ and maps $i_X\co
A\rightarrow X$ and $j_X\co B\rightarrow X$ and $p_X \co C\rightarrow
X^{\II}$ such that $p_X~z~0 = i_X~(u~z)$ and $p_X~z~1 = j_X~(v~z)$. As
above, there is a canonical notion of $D$-algebra maps, and (in
suitable presheaf models) we have an initial $D$-algebra, which we
write $\po(D) = A\sqcup_C B,\hcomp,\inl,\inr,\push$.

We can relativize this situation over a type $\Gamma$.  If $A,B,C$ are
families of types over $\Gamma$ and $u$ (resp.\ $v$) is a family of
maps $u\rho \co C\rho\rightarrow A\rho$ (resp.\ $v\rho \co
C\rho\rightarrow B\rho$), we consider $D = (C,A,B,u,v)$ to be a span
over $\Gamma$, with $D\rho = (C\rho,A\rho,B\rho,u\rho,v\rho)$.  If the
span $D$ is given over $\Gamma$, we define $\po(D)$ in a pointwise way
as for the suspensions, taking $\po(D)\rho$ to be $\po(D\rho)$.

 We want to prove that if $C,A,B$ have transport structures,
then so does $\po(D)$. In order to do that, we first show how to define
further operations from a given transport structure.

\begin{lemma}\label{squeeze}
  Given a family of types $A$ over $\Gamma$ with a transport structure
  $t_A$ we can define a new operation $f_A$ such that
  $f_A~\varphi~\gamma~a_0$ is a path in $\Pi (i:\II)A\gamma(i)$
  constant on $\varphi$ and connecting $a_0$ to
  $t_A~\gamma~\varphi~a_0$ for any $\gamma$ in $\Gamma^{\II}$ constant
  on $\varphi$ and $a_0$ in $A\gamma(0)$. Furthermore given any $a$ in
  $\Pi (i:\II)A\gamma(i)$ we can define an operation
  $sq_A~\varphi~\gamma~a$ which is a path in $(A\gamma(1))^{\II}$
  connecting $t_A~\gamma~\varphi~a(0)$ to $a(1)$, and which is
  constant on $\varphi$.
\end{lemma}
\begin{proof}
  We define
  \[
    f_A~\varphi~\gamma~a_0 = \lambda (i:\II)\, t_A~(\lambda (j:\II)\gamma(i\wedge j))~(\varphi\vee
    (i=0))~a_0
  \]
  which connects $a_0$ to $t_A~\gamma~\varphi~a_0$ and is constant on
  $\varphi$, and
  \[
    sq_A~\varphi~\gamma~a = \lambda (i:\II)\, t_A~(\lambda (j:\II)\gamma(i\vee j))~(\varphi\vee (i=1))~a(i)
  \]
  which connects $t_A~\gamma~\varphi~a(0)$ to $a(1)$ and is constant
  on $\varphi$.
\end{proof}

The relationship between these operations can be displayed as:
\begin{mathpar}
  \begin{tikzpicture}[xscale=5,yscale=2]
    \node (g0) at (0,-0.4) {$\gamma(0)$};
    \node (g1) at (1,-0.4) {$\gamma(1)$};
    \node (A11) at (1,1) {$a(1)$};
    \node (A00) at (0,0) {$a(0)$};
    \node (A10) at (1,0) {$t_A~\gamma~\varphi~a(0)$};
    \path[->,font=\scriptsize,>=angle 90]
      (A00) edge node[above left]{$a$} (A11)
      (A10) edge node[right]{$sq_A~\varphi~\gamma~a$} (A11)
      (A00) edge node[below]{$f_A~\varphi~\gamma~a(0)$} (A10)
      (g0) edge node[below]{$\gamma$} (g1);
    \end{tikzpicture}
\end{mathpar}
so that $sq_A$ can be though of as an operation which ``squeezes'' the
path $a$ into the fiber over $\gamma(1)$.

\begin{corollary}
  Given two families of types $C$ and $A$ over $\Gamma$ with transport
  structures $t_C$ and $t_A$ respectively, and a map $u:C\rightarrow
  A$ over $\Gamma$, there exists an operation $l~\varphi~\gamma~c_0$
  which is a path in $(A\gamma(1))^{\II}$ constant over $\varphi$ and
  connecting $t_A~\gamma~\varphi~(u\gamma(0)~c_0)$ and
  $u\gamma(1)~(t_C~\gamma~\varphi~c_0)$, given $\gamma$ in
  $\Gamma^{\II}$ constant over $\varphi$ and $c_0$ in $C\gamma(0)$.
\end{corollary}
\begin{proof}
  We apply the $sq_A$ operation and the $f_C$ operation from Lemma
  \ref{squeeze} to the path $\lambda (i:\II)\, u\gamma(i)\,
  (f_C~\varphi~\gamma~c_0~i).$
\end{proof}

\begin{proposition}
  Given a family of spans $D = (C,A,B,u,v)$ over a type $\Gamma$ such
  that $A$, $B$, and $C$ have transport structures then the family
  $\po(D)$ also has a transport structure, and hence also a
  composition structure by Lemma~\ref{lemma:compfromtransp}.
\end{proposition}
\begin{proof}
  We use the previous corollary to provide a structure of
  $D\gamma(0)$-algebra on $\po(D)\gamma(1)$, structure which coincides
  with the one of $\po(D)\gamma(0)$ on $\varphi$. By initiality we get
  a map $\po(D)\gamma(0)\rightarrow \po(D)\gamma(1)$ which is the
  identity on $\varphi$, and is the desired transport function. (For a
  more detailed explanation see the syntactic presentation in
  Section~\ref{subsec:pushouts}.)
\end{proof}

\subsection{Existence of initial algebras}

We now explain the proof of Proposition \ref{init1} asserting the
existence of a suitable initial algebra.  We cannot prove this in an
abstract way, but we need to use the fact that we are working with
presheaf models over a small base category $\CC$, in our case the
Lawvere theory of the theory of De~Morgan algebras.  We write
$I,J,K,\dots$ for the objects of $\CC$.  We only describe the case of
$S^1$-algebra here, but all other cases follow the same pattern. The
interested reader may consult Appendix~\ref{sec:appendix1} for the
proofs for the other higher inductive types. The argument we give can
be seen as a constructive version of the small object argument
\cite{Swan14}, and it crucially uses the fact that both $\FF(I)$ and
$\II(I)$ have decidable equality.  Classically we could use Garner's
small object argument \cite{Garner09} as is for instance done in
\cite{LumsdaineShulman}.

We first define inductively a family of sets $\Sp^1_{\text{pre}}(I)$ which is an
``upper approximation'' of the circle, together with maps
$\Sp^1_{\text{pre}}(I) \to \Sp^1_{\text{pre}}(J)$, $u \mapsto uf$ for
$f \co J \to I$. An element of $\Sp^1_{\text{pre}}(I)$ is of the form:
\begin{itemize}
  \item $\base$, or
  \item $\LOOP~r$ with $r\neq 0,1$ in $\II(I)$, or
  \item $\hcomp~[\phi\mapsto u]~u_0$ with $\phi\neq 1$ in $\FF(I)$ and
    $u_0$ in $\Sp^1_{\text{pre}}(I)$ and $u$ a family of elements
    $u_{f,r}$ in $\Sp^1_{\text{pre}}(J)$ for $f:J\rightarrow I$ such
    that $\phi f = 1$ and $r$ in $\II(J)$ .
\end{itemize}

In this way an element of $\Sp^1_{\text{pre}}(I)$ can be seen as a
well-founded tree. Note that we do not yet require that the sides in
$\hcomp$ match up with the base. In order to express this we first
define $uf$ in $\Sp^1_{\text{pre}}(J)$ for $f \co J \to I$ by
induction on $u$:
\begin{align*}
 \base f =& ~ \base \\
 (\LOOP~r) f =&
  \begin{cases}
   \LOOP~(rf) & \text{if } r f \neq 0 \text{ and } rf \neq 1 \\
   \base      & \text{otherwise}
 \end{cases} \\
  (\hcomp~[\phi\mapsto u]~u_0)f =&
  \begin{cases}
   u_{f,1} & \text{if } \phi f = 1 \\
   \hcomp~[\phi f\mapsto uf^+]~(u_0 f) & \text{otherwise}
 \end{cases}
\end{align*}
where $uf^+$ is the family $(uf^+)_{g,r} = u_{fg,r}$ for $g \co K \to
J$.

Note that we may not have in general $(vf)g = v(fg)$ for $v$ in
$\Sp^1_{\text{pre}}(I)$ and $f:J\to I$ and $g:K\to J$.  We then
inductively define the subsets $\Sp^1(I) \subseteq
\Sp^1_{\text{pre}}(I)$ by taking the elements $\base$, $\LOOP~r$, and
$\hcomp~[\phi\mapsto u]~u_0$ such that $u_0 \in \Sp^1(I)$, $u_{f,r}
\in \Sp^1(J)$, for $f:J\to I$ satisfying $u_0 g = u_{g,0}$ for $g:J\to
I$ and $u_{f,r} g = u_{fg,rg}$ for $f : J \to I$ and $r$ in $\II(J)$
and $g : K \to J$.  This defines a cubical set $\Sp^1$, such that
$\Sp^1(I)$ is a subset of $\Sp^1_{\text{pre}}(I)$ for each $I$.

As defined $\Sp^1$ has a structure of an $S^1$-algebra.
Let us sketch that $\Sp^1$ is also the initial $S^1$-algebra in this
presheaf model.  Note that initiality stated internally is a statement
quantifying over all possible types in a universe, which for
simplicity we did not make explicit.  Unfolding this internal
quantification amounts to constructing (suitably unique) natural
transformations $\elim \co \Sp^1 \to A$ where $A$ is a presheaf over
the category of elements of $\yoneda(I)$ equipped with a homogeneous
composition structure and sections $b$ in $A$ and $l$ in $A^\II$
connecting $b$ to itself; moreover, these natural transformations
$\elim$ should be stable under substitutions $\yoneda (f) \co \yoneda
(J) \to \yoneda(I)$.  This works more generally for $A$ being a
presheaf over any cubical set $\Gamma$, not only representables:
$\elim\,\rho\,u$ in $X(I,\rho)$ for $\rho$ in $\Gamma (I)$ and $u$ in
$\Sp^1 (I)$ is defined by induction on the height of the well-founded
tree $u$ simultaneously with verifying $(\elim\,\rho\,u) f =
\elim\,(\rho f)\,(u f)$ for $f \co J \to I$.  Note that the height of
$u f$ does not increase.  Each case in the definition is guided by the
uniqueness condition.

\subsection{Universes}
\label{sec:universes}

As shown externally in \cite{CohenCoquandHuberMortberg15,OrtonPitts16}
(and internally in \cite{LicataEtAl18}) we can define in the presheaf
model a cumulative hierarchy of (univalent and fibrant) universes
$U_n$ which classify families of types of a given size with a
composition structure. Since the way we build initial algebras
preserves the universe level, our definition, e.g., of the suspension
can be seen as an operation $\susp:U_n\rightarrow U_n$.

 Let us expand this point. Let $\mathcal{U}_n$ be a cumulative sequence of
Grothendieck universes (or constructive analog of them \cite{Aczel99}) in the
underlying set theory. If $\Gamma$ is a presheaf on $\CC$
and $A$ a $\mathcal{U}_n$-valued presheaf on the category of elements of $\Gamma$
with a composition structure $c_A$, the suspension operation builds
a $\mathcal{U}_n$-valued presheaf $\susp{A}$ with composition structure
$\susp{c_A}$ such that if $\sigma:\Delta\rightarrow \Gamma$ we have
$(\susp{A})\sigma = \susp {(A\sigma)}$ and
$(\susp{c_A})\sigma  = \susp{(c_A\sigma)}$. An element of $U_n(I)$
is then a pair $A,c_A$ where $A$ is a $\mathcal{U}_n$-valued presheaf on the category
of elements of $\yoneda(I)$ and $c_A$ a composition structure on $A$, and $\susp$
can then be seen as a natural transformation $U_n\rightarrow U_n$.

Thus, we
have presented a semantics of a large class of higher inductive types
with univalent universes. (As shown in \cite{HottBook13}, the
univalence axiom is essential for any non trivial use of the higher-dimensional
structure  of higher inductive types.)

\section{Higher inductive types in cubical type theory}
\label{sec:hits}

In this section we discuss the extensions to cubical type theory by
higher inductive types. We begin by recalling the basic notions of
cubical type theory~\cite{CohenCoquandHuberMortberg15}.

\subsection{Background: cubical type theory}
\label{sec:cubicaltt}

Cubical type theory extends a dependent type theory with a universe
$\UU$ closed under $\Pi$- and $\Sigma$-types with $\Path{}$-types,
composition operations and $\Glue$-types.

The $\Path{}$-types internalize the idea from homotopy type theory
that equalities correspond to paths. We write
$\Path{}\,A\,a\,b$ for the type of paths in $A$ with endpoints $a$ and
$b$. These types behave like function types and have both abstraction
(written $\nabs i \, t$ for $t$ with $i$ abstracted) and application
(written using juxtaposition). The path abstraction binds ``dimension
variables'' ranging over an abstract interval $\II$ specified by the
grammar:
\[
\begin{array}{lcl}
  r,s & ::= & 0 \mid 1 \mid i \mid 1 - r \mid r \wedge s \mid r \vee s
\end{array}
\]

The set $\II$ is a De~Morgan algebra with the $1-r$ operation as
De~Morgan involution.  A type in a context with dimension variables
$i_1,\dots,i_n : \II$ should be thought of as an $n$-dimensional cube
and the substitutions $\subst i 0$ and $\subst i 1$ give the faces of
this cube. A substitution $\subst i j$ renames the dimension variable
$i$ in $A$ into $j$ and as there are no injectivity constraints on
these renaming substitutions one can perform substitutions which give
a ``diagonal'' of a cube (i.e., if $A$ is a square depending on $i, j
: \II$, then $A \subst i j$ is a diagonal). The $\wedge$ and $\vee$
operations are called \emph{connections} and provide convenient ways
of building higher dimensional cubes from lower dimensional ones. For
instance, if $A$ is a line depending on $i$, then $A \subst i {i \land
  j}$ is the interior of the square:
\begin{equation*}
  \begin{tikzpicture}[baseline=10,xscale=3,yscale=2]
    \node (A01) at (0,1) {$A \subst i 0 \subst j 1$};
    \node (A11) at (1,1) {$A \subst i 1 \subst j 1$};
    \node (A00) at (0,0) {$A \subst i 0 \subst j 0$};
    \node (A10) at (1,0) {$A \subst i 1 \subst j 0$};
    \node (c) at (0.5,0.5) {$A\subst{i}{i \land j}$};
    \path[->,font=\scriptsize,>=angle 90]
      (A01) edge node[above]{$A \subst i i$} (A11)
      (A00) edge node[left]{$A\subst i 0$} (A01)
      (A10) edge node[right]{$A \subst{i}{j}$} (A11)
      (A00) edge node[below]{$A\subst i 0$} (A10);
  \end{tikzpicture}
  \qquad %
  \quad %
    \begin{tikzpicture}[->, scale=0.75]
      \draw (0,0) -- node [left] {$j$} (0,1);
      \draw (0,0) -- node [below] {$i$} (1,0);
    \end{tikzpicture}
\end{equation*}

The face lattice $\FF$ is a distributive lattice generated by formal
symbols $(i = 0)$ and $(i = 1)$ with the relation $(i = 0) \wedge (i =
1) = 0_\FF$. The elements of the face lattice can be described by the
grammar:
\[
\begin{array}{lcl}
  \varphi,\psi & ::= & 0_\FF \mid 1_\FF \mid (i = 0) \mid (i = 1)
                    \mid \varphi \land \psi \mid \varphi \lor \psi \\
\end{array}
\]
There is a canonical lattice map $\II \rightarrow \FF$ sending $i$ to
$(i=1)$ and $1-i$ to $(i=0)$. We write $(r=1)$ for the image of $r :
\II$ in $\FF$ and we write $(r=0)$ for $((1-r) = 1)$.

The judgment $\Gamma \der \phi : \FF$ says that $\phi$ is a face
formula involving only the dimension variables declared in
$\Gamma$. Given a formula $\phi$ we can \emph{restrict} a context
$\Gamma$ and obtain a new context written $\Gamma,\phi$ (assuming that
$\phi$ only depends on the dimension variables in $\Gamma$). We call
terms and types in such a restricted context \emph{partial}. These
restricted contexts are used for specifying the boundary of higher
dimensional cubes, for example, if $A$ is a line depending on $i$, the
partial type $i : \II, (i = 0) \lor (i = 1) \der A$ is the two
endpoints of $A$. If $\Gamma,\phi \der v : A$, we write $\Gamma \der u
: A[\phi \mapsto v]$ to denote the two judgments:
\begin{mathpar}
  \Gamma \der u : A
  \and
  \Gamma, \phi \der u = v : A
\end{mathpar}

Using this we can express the typing rule for the composition
operations:
\begin{mathpar}
  \inferrule {\Gamma, i : \II \der A \\
              \Gamma \der \phi : \FF \\
              \Gamma, \varphi, i : \II \der u : A \\
              \Gamma \der u_0 : A\subst i 0[ \varphi \mapsto u \subst i 0 ]}
             {\Gamma \der \comp^i~A~[ \varphi \mapsto u ]~u_0 :
              A\subst i 1[ \varphi \mapsto u\subst i 1 ]}
\end{mathpar}
This operation takes a line type $A$, a formula $\phi$, a partial line
term $u$ and a term $u_0$ of type $A \subst i 0$ (note that $i$ may
occur freely in $A$ and $u$). Furthermore we require that $\Gamma,
\phi \der u_0 = u \subst i 0 : A \subst i 0$.  The result is a term in
$A \subst i 1$ such that $\comp^i~A~[ \varphi \mapsto u ]~u_0 =
u\subst i 1$ on $\Gamma,\phi$. The computation rules for the
composition operations are given as judgmental equalities defined by
cases on the type $A$.

The intuition is that $u$ specifies the sides of an open box while
$u_0$ specifies the bottom of the box and the fact that the sides have
to be connected to the bottom is expressed by the equation relating
$u_0$ and $u \subst i 0$. The result of the composition operation is
then the lid of this open box. For example, given paths $p$, $q$, and
$r$ as in:
\begin{mathpar}
  \begin{tikzpicture}[baseline=30,xscale=2,yscale=2]
    \node (A01) at (0,1) {$c$};
    \node (A11) at (1,1) {$d$};
    \node (A00) at (0,0) {$a$};
    \node (A10) at (1,0) {$b$};
    \path[->,font=\scriptsize,>=angle 90]
      (A00) edge node[left]{$q \; i$} (A01)
      (A10) edge node[right]{$r \; i$} (A11)
      (A00) edge node[below]{$p \; j$} (A10);
    \path[->,font=\scriptsize,>=angle 90,dashed]
      (A01) edge node[above]{$$} (A11);
  \end{tikzpicture}
  \and %
  \hbox{\begin{tikzpicture}[->, scale=0.5, baseline=10]
    \draw (0,0) -- node [left] {$i$} (0,1);
    \draw (0,0) -- node [below] {$j$} (1,0);
  \end{tikzpicture}}
\end{mathpar}
the composition $\comp^i~A~[(j = 0) \mapsto q \; i, (j = 1) \mapsto r
\; i]~(p\;j)$ is the dashed line at the top of the
square.\footnote{Note that we are using a notation for the "system"
  $[(i = 0) \mapsto q \; j, (i = 1) \mapsto r \; j]$. Formally this is
  given by the formula $(i = 0) \lor (i = 1)$ and a partial element
  with endpoints $q \; j$ and $r \; j$.} Here $p \; j$ is a line in $A
\subst i 0$ while $q \; i$ and $r \; i$ are lines in $A \subst j 0$
and $A \subst j 1$, respectively. The resulting composition is then a
line in $A \subst i 1$.

The composition operations allows us to define transport from a line
type:
\begin{mathpar}
  \inferrule {\Gamma, i : \II \der A \\
              \Gamma \der u_0 : A\subst i 0}
             {\Gamma \der \transport^i~A~u_0 = \comp^i~A~[]~u_0 : A\subst i 1}
\end{mathpar}

Combined with ``contractibility of singletons'' (which is directly
provable using a connection) we get the induction principle for
$\Path{}$-types, which means that they behave like Martin-Löf's
identity types (modulo the computation rule for the induction
principle which only holds up to a $\Path{}$).

The $\Glue$-types allow us to prove both the univalence axiom and that
the universe has a composition operation, however as they do not play
an important role in this paper we omit them from this introduction to
cubical type theory.

\subsection{A common pattern for higher inductive types}
\label{sec:common-pattern}

All of the examples of higher inductive types that we consider in this
paper follow a common pattern. In this section we sketch this pattern
which can be seen as a first step towards formulating a syntactic
schema for higher inductive types in cubical type theory, however the
precise formulation of this schema and its semantic counterpart is
left as future work.

Each higher inductive type $\gD (\zs:\Ps)$ is specified by a
telescope\footnote{A \emph{telescope} $x_1 : A_1, \dots, x_n : A_n$
  (written as $\xs : \As$) over a context $\Gamma$ is a (possibly
  empty) list of object variable declarations such that
  $\Gamma,\xs : \As$ is a well-formed context, so $\xs : \As$ neither
  contains context restrictions $\Delta,\phi$ nor dimension variables
  $i : \II$.} of parameters $\zs:\Ps$ (over an ambient context
$\Gamma$) and a \emph{list} of constructors $\vec\gc$.  Each $\gc$ in
$\vec\gc$ is specified by the data:
\begin{equation*}
  \gc : (\xs : \As (\zs)) \, %
  \tabs\is \, %
  \gD (\zs) [\phi (\is) \mapsto e (\zs,\xs,\is)]
\end{equation*}
Here the telescope $\xs:\As$ specifies the types of the arguments to
$\gc$, and in the case of \emph{recursive} higher inductive types, as
in, e.g., propositional truncation, $\gD$ might itself appear in
$\As$.  The length of the list of names $\is$ specifies the dimension
of the cube $\gc$ introduces: we say that $\gc$ is a \emph{point},
\emph{path}, or \emph{square} constructor according to whether the
length of $\is$ is $0,1,$ or $2$, respectively.  The data $\phi
\mapsto e$ specifies the \emph{endpoints} of the constructor $\gc$,
with $\phi$ an element of the face lattice $\FF$ whose free variables
are among $\is$, and $e$ is a partial element
\[
  \zs: \Ps, \xs : \As (\zs), \is : \II, \phi (\is) \der e (\zs,\xs,\is) :
  \gD (\zs)
\]
mentioning only previous constructors in the list $\vec\gc$ and
possibly $\hcomp$'s (see below).

For each instance $\us : \Ps$ of the telescope $\zs : \Ps$ we say that
$\gD (\us)$ is a type and we will have an introduction rule for a
constructor $\gc$ specified as above
\begin{mathpar}
  \inferrule { %
    \vs : \As (\us) \\
    \rs : \II %
  } { \gc\,\vs\,\rs : \gD (\us)}%
\end{mathpar}
and a judgmental equality $\gc\,\vs\,\rs = e (\us,\vs,\rs) : \gD
(\us)$ in case we additionally have $\phi (\rs) = 1 : \FF$ (all in an
ambient context).  Note that this judgmental equality for $\gc$
requires us to make sure that whenever we define a function $\gf : \Pi
(x : \gD (\us))\, P(x)$ that its semantics preserve this equality, so
that
\[
  \phi (\rs) \der \gf (\gc\,\vs\,\rs) = \gf (e (\us,\vs,\rs)) : P
  (\gc\,\vs\,\rs).
\]
In particular, this requirement has to be taken care of in the typing
rules for the eliminator for $\gD (\us)$. The general formulation of
this is left as future work as it would require us to extend cubical
type theory with something similar to the "extension types"
of~\cite{RiehlShulman17}.

Recall from Section~\ref{sec:semantics} that we
decomposed the composition structure for higher inductive types into a
homogeneous composition structure and a transport structure.  The
homogeneous composition structure was introduced as constructors and
the same is reflected in the syntax by adding a rule
\begin{mathpar}
  \inferrule {%
    \Gamma \der \us : \Ps \\
    \Gamma \der \phi : \FF\\\\
    \Gamma, i : \II, \phi \der v : \gD (\us) \\
    \Gamma \der v_0 : \gD (\us) [ \phi \mapsto v \subst i 0 ]
  } %
  { \Gamma \der \hcomp^i_{\gD (\us)} \, [ \phi \mapsto v] \, v_0 : \gD
    (\us) [\phi \mapsto v \subst i 1] }
\end{mathpar}
where the key point is that $i$ may be free in $v$, but \emph{not} in
$D(\us)$, as opposed to the composition operations where $i$ may be
free in both $v$ and $D(\us)$. In the examples we will not repeat
these homogeneous composition constructors for every higher inductive
type we consider and they are always assumed to be included as part of
the definition of the higher inductive type under consideration.

We could do the same for traditional inductive types like the natural
numbers and have a constructor $\hcomp^i_{\N}$ instead of explaining
composition by recursion.  We can prove that this ``weaker'' form of
natural numbers type is equivalent, and hence equal (by univalence) to
the regular one.

To reflect the transport structure in the syntax we specify a
$\transp$ operation for higher inductive types $A := \gD (\us)$ given
$\Gamma, i : \II \der \us : \Ps$ by the rule:
\begin{mathpar}
  \inferrule %
  { \Gamma \der \phi : \FF \\
    \Gamma, i : \II, \phi \der A = A \subst i 0 \\
    \Gamma \der u_0 : A \subst i 0
  }%
  { \Gamma \der \transp^i\,A\,\phi\,u_0 : A \subst i 1 [\phi \mapsto u_0] }%
\end{mathpar}
Note that since $\Gamma, i : \II, \phi \der A = A \subst i 0$ also
$\Gamma, \phi \der A \subst i 0 = A \subst i 1$ (and hence this
equation also holds in context $\Gamma, i : \II, \phi$).

Similar to how the transport structure is explained in the semantics
by recursion on the argument we will add a judgmental equality for
each of the possible shapes of $u_0$: one for each constructor $\gc$
and one for the $\hcomp$ constructor:
\begin{multline*}
  \transp^iA\,\phi\,(\hcomp^j_{A\subst i 0}\,[\psi \mapsto u]\,u_0) =
  \\\hcomp^j_{A \subst i 1}\,[\psi \mapsto
  \transp^iA\,\phi\,u]\,(\transp^iA\,\phi\,u_0)
\end{multline*}
(Note that we can assume that $i \neq j$ as we can always rename one
of them as they are both bound.) As the $\hcomp$ case is the same for
all examples we omit it from the definition of $\transp$ for the
higher inductive types considered in Section~\ref{sec:examples}.

We can define a derived ``$\squeeze$'' operation analogous to $sq_A$
in the proof of Lemma~\ref{squeeze}:
\begin{mathpar}
  \inferrule %
  { \Gamma \der \phi : \FF \\
    \Gamma, i : \II,\phi \der A = A \subst i 0 \\
    \Gamma, i : \II \der a : A } %
  { \Gamma, i : \II \der \squeeze^i\,A\,\phi\,a := %
    \transp^j\,A\subst i {i \lor j}\,(\phi \lor (i=1))\,a : A \subst i 1 } %
\end{mathpar}
This operation satisfies
\begin{align*}
  (\squeeze^i\,A\,\phi\,a) \subst i 0 &= \transp^j\,A\subst i j\,\phi\,a \subst i 0\\
  (\squeeze^i\,A\,\phi\,a) \subst i 1 &= a \subst i 1
\end{align*}
and the induced path is constantly $a$ on $\phi$.

Assuming that we have defined $\transp$ for a higher inductive type
$\Gamma, i : \II \der A$ we can define the composition operation:
\begin{mathpar}
  \inferrule %
  {
    \Gamma \der \phi : \FF \\
    \Gamma, i : \II,\phi \der u : A\\
    \Gamma \der u_0 : A \subst i 0 [\phi \mapsto u \subst i 0 ] } %
  { \Gamma \der \comp^i\,A\,[\phi \mapsto u]\,u_0 := \\%
    \hcomp^i_{A \subst i 1}\,[ \phi \mapsto \squeeze^i\,A\,0_\FF\,u]
    \,(\transp^i\,A\,0_\FF\,u_0) : A \subst i 1 }
\end{mathpar}
This satisfies the required judgmental computation rule for $\comp$
because of the computation rules for $\hcomp$ and $\squeeze$. This
means that in order to define the composition operation for a higher
inductive type we only need to define the $\transp$ operation when
applied to constructors.

\medskip

Note, that we can always define a $\transp$ operation for any type
$\Gamma, i : \II \der A$ that already has a composition operation by:
\begin{mathpar}
  \inferrule %
  { \Gamma \der \phi : \FF \\
    \Gamma, i : \II,\phi \der A = A \subst i 0\\
    \Gamma \der u_0 : A \subst i 0 }%
  { \Gamma \der \ctransp^i\,A\,\phi\,u_0 := %
    \comp^i\,A\,[ \phi \mapsto u_0 ]\,u_0 : A \subst i 1 [\phi \mapsto
    u_0] }%
\end{mathpar}

In line with Lemma~\ref{squeeze} a corresponding ``filling'' operation
which connects the input of $\transp$ to its output can also be
derived:
\begin{mathpar}
  \inferrule %
  { \Gamma \der \phi : \FF \\
    \Gamma, i : \II, \phi \der A = A \subst i 0 \\
    \Gamma \der u_0 : A \subst i 0
  }%
  { \Gamma,i : \II \der \Transp^i\,A\,\phi\,u_0 := %
    \transp^j\,A \subst i {i \land j}\,(\phi \lor (i=0))\,u_0: A}%
\end{mathpar}
Note that $\Gamma,i:\II,\phi \der A = A \subst i 0$ entails
\[
  \Gamma, i : \II, j : \II, \phi \lor (i = 0) \der A \subst i {i \land
    j} = A \subst i {i \land j} \subst j 0.
\]
This operation satisfies
\begin{align*}
  (\Transp^i\,A\,\phi\,u_0) \subst i 0 &= u_0\\
  (\Transp^i\,A\,\phi\,u_0) \subst i 1 &=  \transp^j\,A\subst i j\,\phi\,u_0
\end{align*}
and the induced path is constantly $u_0$ on $\phi$. We write
$\cTransp$ for the corresponding operation defined using $\ctransp$.

\subsection{Examples of higher inductive types}
\label{sec:examples}

In this section we describe how to extend cubical type theory with the
circle and spheres, torus, suspensions, propositional truncation, and
pushouts.  As with all the other type formers we have to explain their
formation, introduction, elimination, and computation rules, as well
as how composition computes.  All of these examples follow the common
pattern presented in the previous section.

\subsubsection{The circle and spheres}
\label{subsec:spheres}

The extension of cubical type theory with the circle and spheres was
sketched in~\cite[Section~9.2]{CohenCoquandHuberMortberg15} and we
elaborate on this here.

\paragraph{Formation}

In order to extend the theory with the circle we first add it as a
type by:
\begin{mathpar}
  \inferrule { {} } {\der \Sp^1} %
  \and
  \inferrule { {} } {\Sp^1 : \UU} %
\end{mathpar}

\paragraph{Introduction}

The circle is generated by a point and a path constructor:
\begin{mathpar}
  \inferrule { {} } {\base:\Sp^1} %
  \and %
  \inferrule {r : \II} {\LOOP\,r : \Sp^1}
\end{mathpar}
with the judgmental equalities $\LOOP\,0 = \LOOP\,1 = \base$ so that
$\LOOP$ connects the point to itself.

\paragraph{Elimination}
Given a dependent type $x : \Sp^1 \der P(x)$, a term $b : P (\base)$
and a path $i : \II \der l : P (\LOOP\,i) [(i=0) \lor (i=1) \mapsto
\,b]$ we can define $\gf : \Pi (x : \Sp^1)\, P (x)$ by cases:
\begin{mathpar}
  \gf\,\base = b %
  \\
  \gf\,(\LOOP\,r) = l\,r
\end{mathpar}
and for the $\hcomp$ constructor:
\[
  \gf\,(\hcomp^i_{\Sp^1}\,[\phi \mapsto u]\,u_0) = \comp^i\,P
  (v)\,[\phi \mapsto \gf\,u]\,(\gf\,u_0)
\]
where w.l.o.g.\ $i$ is fresh and:
\begin{align*}
  v &:= \hComp^i\,\Sp^1\,[\phi \mapsto u ]\,u_0
  \\ %
    &\phantom{:}= \hcomp^j_{\Sp^1} [\phi \mapsto u \subst i {i \land j}, (i=0)
      \mapsto u_0]\,u_0
\end{align*}

As the equation for the eliminator applied to an $\hcomp$ is analogous
for all the other higher inductive considered here we will omit it in
the sequel.

Using this we can directly define the eliminator:
\begin{mathpar}
  \inferrule {x : \Sp^1 \der P(x) \\
              b : P(\base) \\
              l : \Path^i\,P(\LOOP\,i)\,b\,b \\
              u : \Sp^1}
            { \Spelim{1}{x.P}\,b\,l\,u : P(u)}
\end{mathpar}
where $\Path^i{}$ denotes a dependent path type
(see~\cite[Section~9.2]{CohenCoquandHuberMortberg15}). The judgmental
computation rules then follow from the definition above. Note that as we have
dependent $\Path{}$-types (which behave like heterogeneous equalities)
the $\LOOP$ case of $\gf$ can be expressed directly by an equation
without ``\textsf{apd}'' and $l$ does not involve any $\transport$ as
opposed to~\cite[Section 6.4]{HottBook13}.

\paragraph{Composition}
As $\Sp^1$ has no parameters we let $\transp^i\,\Sp^1\,\phi\,u_0 =
u_0$. This means that the composition $\comp^i\,\Sp^1\,[\phi \mapsto
u]\,u_0$ computes directly to the constructor $\hcomp^i_{\Sp^1}\,[\phi
\mapsto u]\,u_0$.

\medskip

The higher dimensional spheres, $\Sp^n$, can directly be defined by
generalizing the definition $\Sp^1$ so that $\LOOP$ takes
$r_1,\dots,r_n : \II$. It is trivial to define
$\transp^i\,\Sp^n\,\phi\,u_0$ in analogy with $\Sp^1$. The elimination
is also analogous to that of $\Sp^1$ using an $n$-dimensional cube in
$P (\LOOP\,i_1\,\dots\,i_n)$ for the $\LOOP$ case.

\subsubsection{The torus; two equivalent formulations}
\label{subsec:torus}

We define the torus in two ways, the first one (written $\TT$) is
analogous to $\Sp^2$ and the second (written $\TTF$) is the cubical
analogue of the torus as defined in~\cite[Section
6.6]{HottBook13}. The $\TTF$ torus involves the fibrancy structure of
the $1$-dimensional cells in the $2$-dimensional cell. Higher
inductive types of this kind are not supported
by~\cite{LumsdaineShulman} and we make crucial use of the fact that we
have homogeneous composition as a constructor in order to represent
them.

\paragraph{Formation}
The formation rules for the torus types are given by:
\begin{mathpar}
  \inferrule { {} } {\der \TT} %
  \and
  \inferrule { {} } {\TT : \UU} %
  \and
  \inferrule { {} } {\der \TTF} %
  \and
  \inferrule { {} } {\TTF : \UU} %
\end{mathpar}

\paragraph{Introduction}
The point, lines and square constructors for $\TT$ are given by:
\begin{mathpar}
  \inferrule { {} } {\tb : \TT}
  \and %
  \inferrule {r : \II} {\tp\,r : \TT}
  \and %
  \inferrule {r : \II} {\tq\,r : \TT}
  \and %
  \inferrule {r : \II \\ s : \II}
             {\tsurf\,r\,s : \TT }
\end{mathpar}
satisfying $\tp\,0 = \tp\,1 = \tq\,0 = \tq\,1 = \tb$. The constructors
for $\TTF$ are defined by the same rules as for $\TT$ and we write
them with $\textsf{F}$ as subscript. The square constructor for $\TT$
satisfies $\tsurf\,0\,s = \tsurf\,1\,s = \tp\,s$ and $\tsurf\,r\,0 =
\tsurf\,r\,1 = \tq\,r$ so that we get the square representing the traditional
gluing diagram used in the topological definition of the torus:
\begin{mathpar}
  \begin{tikzpicture}[baseline=30,xscale=2,yscale=2]
    \node (A01) at (0,1) {$\tb$};
    \node (A11) at (1,1) {$\tb$};
    \node (A00) at (0,0) {$\tb$};
    \node (A10) at (1,0) {$\tb$};
    \node (C) at (0.5,0.5) {$\tsurf\,i\,j$};
    \path[->,>=angle 90]
      (A01) edge node[above]{$\tq\,i$} (A11)
      (A00) edge node[left]{$\tp\,j$} (A01)
      (A10) edge node[right]{$\tp\,j$} (A11)
      (A00) edge node[below]{$\tq\,i$} (A10);
  \end{tikzpicture}
  \and %
  \vcenter{\hbox{\begin{tikzpicture}[->, scale=0.75]
    \draw (0,0) -- node [left] {$j$} (0,1);
    \draw (0,0) -- node [below] {$i$} (1,0);
  \end{tikzpicture}}}
\end{mathpar}

Given $s : \II$ we define the composition of $\tfp$ and $\tfq$ by:
\[
  \tfp \cdot_s \tfq := %
  \hcomp^i_{\TTF}\,[(s=0) \mapsto \tfb,(s=1) \mapsto
  \tfq\,i]\,(\tfp\,s)
\]
The composition $\tfq \cdot_s \tfp$ is defined analogously.

The square constructor for $\TTF$ satisfies $\tfsurf\,0\,s =
\tfp\,\cdot_s\,\tfq$, $\tfsurf\,1\,s = \tfq\,\cdot_s\,\tfp$ and
$\tfsurf\,r\,0 = \tfsurf\,r\,1 = \tfb$. This way the $2$-cell $\nabs
{i\,j} \, \tfsurf\,i\,j$ corresponds to a cubical version of the globe
(which can be turned into a square with reflexivity at $\tfb$ as
sides):
\begin{center}
  \begin{tikzpicture}[xscale=3,yscale=2]
    \node (A00) at (0,0) {$\tfb$};
    \node (A10) at (1,0) {$\tfb$};
    \path[->,>=angle 90]
      (A00) edge[bend left=45] node[above]{$\tfp \cdot_j \tfq$} (A10)
      (A00) edge[bend right=45] node[below]{$\tfq \cdot_j \tfp$} (A10);
  \end{tikzpicture}
\end{center}

\paragraph{Elimination}
We write $(i=0/1)$ for $(i = 0) \lor (i=1)$. Given a dependent type $x
: \TT \der P (x)$, a term $b : P (\tb)$, paths $i : \II \der l_p : P
(\tp\,i)[(i=0/1) \mapsto \,b]$ and $i : \II \der l_q : P (\tq\,i)
[(i=0/1) \mapsto \,b]$ and a square $i, j : \II \der s_{pq} :
P(\tsurf\,i\,j)\,[(i=0/1) \mapsto l_p\,j, (j=0/1) \mapsto l_q\,i]$ we
can define $\gf : \Pi(x : \TT)\, P (x)$ by cases:
\begin{align*}
  \gf\,\tb &= b \\
  \gf\,(\tp\,r) &= l_p\,r \\
  \gf\,(\tq\,r) &= l_q\,r \\
  \gf\,(\tsurf\,r\,s) &= s_{pq}\,r\,s 
\end{align*}

Similarly for a dependent type $x : \TTF \der P (x)$, a term $b : P
(\tfb)$, paths $i : \II \der l_p : P (\tfp\,i)[(i=0/1) \mapsto \,b]$
and $i : \II \der l_q : P (\tfq\,i) [(i=0/1) \mapsto \,b]$ we define:
\[
  l_p \cdot_j l_q :=%
  \comp^i\,P(v)\,[(j=0) \mapsto b,(j=1) \mapsto l_q\,i]\,(l_p\,j)
\]
where $v := \hComp^i_{\TTF}\,[(j=0) \mapsto \tfb,(j=1) \mapsto
\tfq\,i]\,(\tfp\,j)$. We define $l_q \cdot_j l_p$ analogously and we
can then require a square $i, j : \II \der s_{pq} : P(\tfsurf\,i\,j)\,
[ (i=0) \mapsto l_p \cdot_j l_q, (i=1) \mapsto l_q \cdot_j l_p,
(j=0/1) \mapsto b ]$. Using this we can define $\gf : \Pi(x : \TTF)\,
P (x)$ by cases like for $\TT$.

Working with $\TT$ is easier than $\TTF$ and the proof that $\TT
\simeq \Sp^1 \times \Sp^1$ has been formalized in \cubicaltt{} by Dan
Licata.\footnote{See:
  \url{https://github.com/mortberg/cubicaltt/blob/hcomptrans/examples/torus.ctt}}
The proof of this is very direct and a lot shorter than the existing
proofs in the literature~\cite{Sojakova16Torus,LicataBrunerie15}. One
first defines maps $f_1 : \TT \to \Sp^1 \times \Sp^1$ and $f_2 :
\Sp^1 \times \Sp^1 \to \TT$ by:
\begin{align*}
  f_1\,\tb            &= (\base,\base)       & f_2\,(\base,\base)       &= \tb \\
  f_1\,(\tp\,r)       &= (\LOOP\,r,\base)    & f_2\,(\LOOP\,r,\base)    &= \tp\,r \\
  f_1\,(\tq\,r)       &= (\base,\LOOP\,r)    & f_2\,(\base,\LOOP\,r)    &= \tq\,r \\
  f_1\,(\tsurf\,r\,s) &= (\LOOP\,r,\LOOP\,s) & f_2\,(\LOOP\,r,\LOOP\,s) &= \tsurf\,r\,s
\end{align*}

These are obviously inverses and the equivalence can be
established. The formal proof in \cubicaltt{} is slightly more
complicated as it is not possible to directly do the double recursion
in $f_2$, but the basic idea is the same. This example shows how
having a system where higher inductive types compute also for
higher constructors makes it possible to simplify formal proofs in
synthetic homotopy theory.

\paragraph{Composition}
As neither $\TT$ or $\TTF$ have any parameters the transport operation
is trivial just like for $\Sp^n$, so the composition operations
reduces to the $\hcomp$ constructors.

\subsubsection{Suspension}
\label{subsec:suspension}

The suspension of a type $A$, written $\susp{A}$, is more involved
than the higher inductive types considered so far as it has a
parameter and just as in the semantics we have to explain the
transport operation.

\paragraph{Formation}
In order to extend the theory with suspensions we add the rules:
\begin{mathpar}
  \inferrule {\der A} {\der \susp{A}} %
  \and
  \inferrule {A : \UU} {\susp{A} : \UU} %
\end{mathpar}

Note that we allow $\susp{A}$ to be in the same universe as $A$, this
is justified by the semantics as explained in
Section~\ref{sec:universes}.

\paragraph{Introduction}
The suspensions are generated by:
\begin{mathpar}
  \inferrule { {} } {\north : \susp{A}} %
  \and %
  \inferrule { {} } {\south : \susp{A}} %
  \and %
  \inferrule {a : A \\
              r : \II}
             {\merid\,a\,r : \susp{A}}
\end{mathpar}
satisfying $\merid\,a\,0 = \north$ and $\merid\,a\,1 = \south$.

\paragraph{Elimination}
Given a dependent type $x : \susp{A} \der P(x)$, terms $n : P(\north)$
and $s : P(\south)$ and a family of paths $x : A, i : \II \der m(x,i)
: P(\merid\,x\,i)[(i=0) \mapsto n,(i=1) \mapsto s]$ we can define a
function $\gf : \Pi(x : \susp{A})\, P(x)$ by cases:
\begin{align*}
  \gf\,\north &= n \\
  \gf\,\south &= s \\
  \gf\,(\merid\,a\,r) &= m(a,r) 
\end{align*}

\paragraph{Composition}
The $\transp^i(\susp{A})\,\phi\,u_0$ operation is defined as
\begin{align*}
  \transp^i(\susp{A})\,\phi\,\north &= \north \\
  \transp^i(\susp{A})\,\phi\,\south &= \south \\
  \transp^i(\susp{A})\,\phi\,(\merid\,a\,r) &= \merid\,(\ctransp^i\,A\,\phi\,a)\,r
\end{align*}

\subsubsection{Propositional truncations}
\label{subsec:truncations}

Another class of interesting higher inductive types are the
truncations; these introduce some new complications as they are
recursive in the sense that the higher constructors quantify over
elements of the type.  The propositional truncation takes a type $A$
and ``squashes'' it to a $0$-type $\inh{A}$ (in the sense that the
equality type of $\inh{A}$ has no interesting structure).

\paragraph{Formation}
In order to extend the theory with propositional truncation we add the
rules:
\begin{mathpar}
  \inferrule {\der A} {\der \inh{A}} %
  \and
  \inferrule {A : \UU} {\inh{A} : \UU} %
\end{mathpar}

\paragraph{Introduction}
The propositional truncation of $A$ is generated by:
\begin{mathpar}
  \inferrule {a : A} %
             {\inc a : \inh A} %
  \and %
  \inferrule {v : \inh A \\
              w : \inh A \\
              r : \II} %
             {\squash\,v\,w\,r : \inh A} %
\end{mathpar}
satisfying $\squash\,v\,w\,0 = v$ and $\squash\,v\,w\,1 = w$.

\paragraph{Elimination}
Given a dependent type $x : \inh{A} \der P(x)$, a family of terms $x :
A \der t(x) : P(\inc\,x)$ and family of paths $v, w : \inh{A}, x :
P(v), y : P(w), i : \II \der p(v,w,x,y,i) : P(\squash\,v\,w\,i)[(i=0)
\mapsto x,(i=1) \mapsto y]$ we can define $\gf : \Pi (x : \inh{A})\,
P(x)$ by cases:
\begin{align*}
  \gf\,(\inc a) &= t(a) \\
  \gf\,(\squash\,v\,w\,r) &= p(v,w,\gf\,v,\gf\,w,r) 
\end{align*}
This is directly structurally recursive and the only difference
compared to $\susp{A}$ is that we have to make a recursive call for
each recursive argument.

\paragraph{Composition}
We define $\transp^i\inh{A}\,\phi\,u_0$ by cases on $u_0$:
\begin{align*}
  \transp^i\inh{A}\,\phi\,(\inc\,a)
  &= \inc\,(\ctransp^i\,A\,\phi\,a) \\
  \transp^i\inh{A}\,\phi\,(\squash\,v\,w\,r)
  &= \squash\,(\transp^i\inh{A}\,\phi\,v)\,(\transp^i\inh{A}\,\phi\,w)\,r 
\end{align*}

The explanation of propositional truncation
in~\cite[Section~9.2]{CohenCoquandHuberMortberg15} used a similar
decomposition, but the introduction of the $\transp$ operation allows
a much simpler formulation of composition.

\subsubsection{Pushouts}
\label{subsec:pushouts}

The definition of pushouts in cubical type theory is similar to the
other parametrized higher inductive types, but special care has to be
taken when defining $\transp$ as the endpoints of the path
constructors involve the parameters to the pushout.

\paragraph{Formation}
We extend the theory with:
\begin{mathpar}
  \inferrule
  { \der A \\ \der B \\ \der C \\
    \pushf : C \to A \\
    \pushg : C \to B} %
  { \der \Push }
  \and
  \inferrule
  { A : \UU \\ B : \UU \\ C : \UU \\
    \pushf : C \to A \\
    \pushg : C \to B} %
  { \Push : \UU }
\end{mathpar}

\paragraph{Introduction}
Given $\pushf : C \to A$ and $\pushg : C \to B$ the pushout is generated by:
\begin{mathpar}
  \inferrule
  { a : A} %
  { \inl a : \Push } %
  \and %
  \inferrule
  { b : B} %
  { \inr b : \Push } %
  \and %
  \inferrule
  { c : C \\ r : \II } %
  { \push\,c\,r : \Push} %
\end{mathpar}
satisfying $\push\,c\,0 = \inl (\pushf \, c)$ and $\push\,c\,1 = \inr
(\pushg \, c)$. Note that $\nabs i \, {\push\,c\,i}$ gives a path
between $\inl (\pushf\,c)$ and $\inr (\pushg \, c)$ for all $c : C$ as
desired.

\paragraph{Elimination}
Given a dependent type $x : \Push \der P(x)$, families of terms $x : A
\der l(x) : P(\inl\,x)$ and $x : B \der r(x) : P(\inr\,x)$ and a
family of paths $x : C, i : \II \der p(x,i) : P(\push\,x\,i)[(i=0)
\mapsto l(\pushf\,x),(i=1) \mapsto r(\pushg\,x)]$ we can define $\gf :
\Pi(x : \Push)\, P(x)$ by cases:
\begin{align*}
  \gf\,(\inl a) &= l(a) \\
  \gf\,(\inr b) &= r(b) \\
  \gf\,(\push\,c\,r) &= p(c,r) 
\end{align*}

\paragraph{Composition}

We write $P$ for $\Push$ and the judgmental computation rules for
$\transp$ are defined by cases:
\begin{align*}
  \transp^i\,P\,\phi\,(\inl a) &= \inl \, (\ctransp^i\,A\,\phi\,\,a) \\
  \transp^i\,P\,\phi\,(\inr b) &= \inr \, (\ctransp^i\,B\,\phi\,\,b) \\
  \transp^i\,P\,\phi\,(\push\,c\,r) &= \hcomp_{P\subst i 1}^i\,S\,(\push\,(\ctransp^i\,C\,\phi\,c) \,r) 
\end{align*}
where  $S$ is the system:
\begin{align*}
  [ & (r=0) \mapsto \squeeze^i\,P\,\phi\,(\inl
      (\pushf\,(\cTransp^i\,C\,\phi\,c))) \, \subst i {1-i},
  \\
    & (r=1) \mapsto \squeeze^i\,P\,\phi\,(\inr
      (\pushg\,(\cTransp^i\,C\,\phi\,c))) \, \subst i {1-i},
  \\
    & (\phi=1) \mapsto \push\,c\,r ]
\end{align*}
Note that the recursive call to $\squeeze$ is justified as it is
applied to a point constructor which has already been defined.

Furthermore, note that the endpoint correction for $\push\,c\,r$ is
necessary as, for example, in the case where $r$ is a dimension
variable $j$ the path constructor $\push\,(\ctransp^i\,C\,\phi\,c)
\,j$ connects
\begin{mathpar}
  \inl(\pushf \subst i 1 \, (\ctransp^i\,C\,\phi\,c))
  \and
  \text{to}
  \and
  \inr(\pushg \subst i 1 \, (\ctransp^i\,C\,\phi\,c))
\end{mathpar}
in direction $j$, but we require something that connects
\begin{mathpar}
  \inl (\ctransp^i\,A\,\phi\,(\pushf \subst i 0 \,c))
  \and
  \text{to}
  \and
  \inr (\ctransp^i\,B\,\phi\,(\pushg\subst i 0\,c))
\end{mathpar}
since the definition of $\transp$ should be stable under the
substitutions $\subst j 0$ and $\subst j 1$. To see that the
correction is correct at $(r=0)$ note that $\squeeze^i\,P\,\phi\,(\inl
(\pushf\,(\cTransp^i\,C\,\phi\,c))) \, \subst i {1-i}$ connects
\begin{mathpar}
  \inl (\pushf \subst i 1\,(\ctransp^i\,C\,\phi\,c))
  \and
  \text{to}
  \and
  \inl (\ctransp^i\,A\,\phi\,(\pushf\subst i 0 \,c))
\end{mathpar}
as required.

\subsection{A variation on cubical type theory}
\label{sec:cubicaltt-variant}

In the previous section we have seen that the equations to define
$\transp^i\,A$ for a higher inductive type $A$ applied to a
constructor involves $\transp^i\,A$ for the recursive arguments to the
constructor (see the equation for $\squash\,v\,w\,r$ for propositional
truncation in Section~\ref{subsec:truncations}), and involves the
derived operations $\ctransp$ for non-recursive arguments (e.g., in
the equation for $\merid\,a\,r$ in Section~\ref{subsec:suspension}).
In general, $\transp$ and $\ctransp$ which are available for $A$
do not coincide definitionally, making it impossible to treat
the recursive and non-recursive arguments to a constructor uniformly.

This mismatch suggests a variant of cubical type theory where the
operations $\transp$ and $\hcomp$ are taken as primitives and
$\comp$ is instead a derived operation as we did here for higher
inductive types.  We can then explain $\transp$ and $\hcomp$ by cases
on the shape of the type.  In this variation of cubical type theory
the algorithm for $\transp$ in a higher
inductive type applied to a constructor can be uniformly described as follows.

Given a higher inductive type $\gD (\zs : \Ps)$ specified as in
Section~\ref{sec:common-pattern} and a constructor $\gc$ specified by:
\begin{equation}
  \label{eq:generic-constr}
  \gc : (\xs : \As (\zs)) \, %
  \tabs\is \, %
  \gD (\zs) [\phi (\is) \mapsto e (\zs,\xs,\is)]
\end{equation}
Further, assume parameters $\Gamma, i : \II \der \us : \Ps$ of the
higher inductive type $\gD (\zs : \Ps)$ such that $\Gamma, i : \II,
\psi \der \us = \us \subst i 0 : \Ps$ for $\Gamma \der \psi : \FF$.
We now explain the judgmental equalities of
\[
  w_1 := \transp^i\,\gD (\us)\,\psi\,(\gc\,\vs\,\rs)
\]
for $\Gamma \der \vs : \As (\us \subst i 0)$ and $\Gamma \der \rs :
\II$.  This $\gc\,\vs\,\rs$ restricts to $\phi (\rs) \mapsto e (\us
\subst i 0, \vs, \rs)$.  We want to define $\Gamma \der w_1 : \gD (\us
\subst i 1) [\psi \mapsto \gc\,\vs\,\rs]$ such that $w_1$ restricts to
\begin{equation}
  \label{eq:required-extent-w1}
  \phi (\rs) \mapsto
  \transp^i\,\gD (\us)\,\psi\,e (\us \subst i 0, \vs, \rs).
\end{equation}

We get a line in $\xs : \As (\us)$ in the context $\Gamma, i : \II$
\[
  \begin{tikzpicture}[xscale=6]
    \node (VS) at (0,0) {$\vs$}; %
    \node (VStrans) at (1,0) {$\transp^i\,(\xs:\As
      (\us))\,\psi\,\vs$}; %
    \draw[->] (VS) to %
    node[above] {$\thetas := \Transp^i\,(\xs:\As
      (\us))\,\psi\,\vs$} (VStrans);
  \end{tikzpicture}
\]
along $i$, where $\Transp^i\,(\xs : \As)\,\psi\,\vs$ is the extension
of $\Transp$ to telescopes, mapping the empty telescope to itself, and
\[
  \Transp^i\,(x:A,\xs : \As(x))\,\psi\,(v,\vs) = \tilde
  v,\Transp^i\,(\xs : \As(\tilde v))\,\psi\,\vs
\]
with $\tilde v = \Transp^i\,A\,\psi\,v$.  The extension of $\transp$
to telescopes is the $\subst i 1$ face of the corresponding $\Transp$.

We start with $\Gamma \der w_1' : \gD (\us \subst i 1)$ given by
\[
  w_1' := \gc\,(\transp^i\,(\xs : \As)\,\psi\,\vs)\,\rs
\]
which restricts to $\phi (\rs) \mapsto e (\us \subst i 1,
\transp^i\,(\xs : \As)\,\psi\,\vs,\rs)$ and which we have to correct
to match~\eqref{eq:required-extent-w1}.  To make this correction,
consider the line $\Gamma, \phi (\rs), i : \II \der \alpha(i) : \gD
(\us \subst i 1)$ given by
\[
  \alpha (i) := \squeeze^i\,\gD(\us)\,\psi\,e(\us, \thetas,\rs)
\]
connecting the element in~\eqref{eq:required-extent-w1} to $e (\us
\subst i 1, \transp^i\,(\xs : \As)\,\psi\,\vs,\rs)$.  Note that
$\alpha (i)$ coincides with $e (\us \subst i 0,\vs,\rs)$ (and hence
with $\gc\,\vs\,\rs$) on $\psi$.

We now add the judgmental equality
\[
  w_1 = \hcomp^i_{\gD (\us \subst i 1)}\,[\phi (\rs) \mapsto \alpha (1
  - i), \psi \mapsto \gc\,\vs\,rs]\,w_1'.
\]

Note that in the definition $\alpha$ we recursively call $\transp$ for
$\gD$ on $e$.  To ensure that this call is well-founded it is crucial
to have restrictions on how $e$ may look like.

Also note that this algorithm might not be optimal: For a higher
inductive type without any parameters (e.g., $\Sp^1$) we could have
simply defined $\transp$ to be the identity as we did in the previous
section. For a type where the endpoints of constructors are suitably
simple, like suspensions and propositional truncation, but not
pushouts, we could have directly taken $w_1'$ above. This has the
consequence that the result might have some unnecessary $\hcomp$'s and
would equal, up to a $\Path$, to a simpler term without these
$\hcomp$'s.

Our general pattern of constructors~\eqref{eq:generic-constr} suggests
to formulate a schema.  Such a schema would have to ensure that
$\gD(\zs)$ only appears strictly positive in $\As$ and would have to
restrict what possible endpoints $e$ are allowed.  We leave the
detailed formulation of the semantics of such a schema as future work.

\section{Conclusions and related work}
\label{sec:conclusions}

In this paper we constructed the semantics of some important higher
inductive types in cubical sets.  A crucial ingredient was the
decomposition of the composition structure into a homogeneous
composition structure and a transport structure.  Using this
decomposition we define higher inductive type formers such that they
preserve the universe level and are strictly stable under
substitution.

We also extended cubical type theory with some higher inductive types.
While~\cite{Huber16} only proves canonicity for cubical type theory
extended with the circle and propositional truncation, it should be
straightforward to extend this result to the higher inductive types
presented in this paper using the obvious operational semantics
obtained by orienting the judgmental equalities given here.  It also
remains to prove normalization and decidability of type-checking for
cubical type theory and in particular also for our extension with
higher inductive types.

As mentioned in Section~\ref{sec:cubicaltt-variant}, it is more
natural for a general treatment of higher inductive types to formulate
a variation of cubical type theory based homogeneous compositions and
transport as primitive instead of heterogeneous compositions.  It
seems that our description of transport for higher inductive types
also works for a more general schema, but its details and semantics
still have to be worked out.

Using the experimental implementation of the system presented in this
paper we have formalized the ``Brunerie number''\footnote{The complete
  self-contained formalization can be found at:
  \url{https://github.com/mortberg/cubicaltt/blob/hcomptrans/examples/brunerie.ctt}},
i.e., $n$ such that $\pi_4(\Sp^3) \simeq \mathbb{Z}/n\mathbb{Z}$. The
formalization closely follows~\cite[Appendix B]{Brunerie16} and the
definition involves multiple higher inductive types (the spheres,
truncations, and join construction) together with many uses of the
univalence axiom.  By the classical definition of this homotopy group
we know that the expected value for $n$ is $\pm 2$ and this also is
proved to be the case in~\cite{Brunerie16}.  But as we have a
constructive justification for all of the notions involved in the
definition we can in principle directly obtain this numeral by
computation.  However, this computation so far has been unfeasible.

Further future work is to relate our semantics to other models of
homotopy type theory.  In particular, clarify the connection of the
model structure on cubical sets~\cite{Sattler17} and the usual model
structure on simplicial sets. It is also of interest to investigate to
what extent the techniques developed in this paper can be adapted to
the simplicial set model.\footnote{See the following discussion for
  more details:
  \url{https://groups.google.com/d/msg/homotopytypetheory/bNHRnGiF5R4/3RYz1YFmBQAJ}}

\paragraph{Related work}

The
papers~\cite{AngiuliEtAl17,CHTTPOPL,AngiuliHouHarper17,CavalloHarper18}
present cubical type theories inspired by an alternative cubical set
category with different fibrancy structure, but with the same
decomposition of the composition operation in a homogeneous
composition and a transport operation. This decomposition was
introduced in an early version of~\cite{CohenCoquandHuberMortberg15}
precisely to solve the problem of the interpretation of higher
inductive types with parameters. The suspensions are covered
in~\cite{AngiuliEtAl17}, and~\cite{CavalloHarper18} defines a schema
for higher inductive types formulated in this setting. The
papers~\cite{CHTTPOPL,AngiuliHouHarper17,CavalloHarper18} describe
computational type theories in the style of Nuprl with a semantics
where types are interpreted as partial equivalence relations which
gives canonicity for booleans. The schema presented
in~\cite{CavalloHarper18} covers all of the examples
of higher inductive types considered in this paper.

The paper \cite{LumsdaineShulman} presents a semantics of higher
inductive types in a general framework of ``sufficiently nice''
Quillen model categories. However as it is now, it models a type
theory which does not contain any universes
(see~\cite[pp.~5--6]{LumsdaineShulman} for a discussion of this
point).

A schema with point, path, and square constructors expressed in the
style of~\cite{HottBook13} is presented in~\cite{Dybjer}. This paper
also contains a semantics for these higher inductive types in the
groupoid model.

\bibliographystyle{plainurl}
\bibliography{thesis}

\newpage
\appendix

\section{Appendix: construction of initial algebras}
\label{sec:appendix1}
In this appendix we sketch how to construct the semantic versions of
the examples of higher inductive types $T$ that we consider. With
suitable definitions of $T$-algebra structures these proofs can be
seen as constructions of initial $T$-algebras.

\paragraph{Torus}

The semantic version of $\TT$ is very similar to that of $\Sp^1$, so
we only give the semantics of $\TTF$ as it is more interesting. Just
as for the circle we first define an upper approximation of sets
$\TTF^{\text{pre}}(I)$, together with maps $\TTF^{\text{pre}}(I) \to
\TTF^{\text{pre}}(J)$ for $f : J \to I$. An element of
$\TTF^{\text{pre}}(I)$ is of the form:
\begin{itemize}
  \item $\tfb$, or
  \item $\tfp\,r$ or $\tfq\,r$ with $r\neq 0,1$ in $\II(I)$, or
  \item $\tfsurf\,r\,s$ with $r,s\neq 0,1$ in $\II(I)$, or
  \item $\hcomp~[\phi\mapsto u]~u_0$ with $\phi \neq 1$ in $\FF(I)$
    and $u_0$ in $\TTF^{\text{pre}}(I)$ and $u$ a family of elements
    $u_{f,r}$ in $\TTF^{\text{pre}}(J)$ for $f : J \to I$ and $r$ in $\II(I)$ such that
    $\phi f = 1$.
\end{itemize}

We write $\tfp \cdot_r \tfq$ for
\[
  \hcomp~[ (r = 0) \mapsto \tfb \, \pp, (r = 1) \mapsto
  \tfq\,\qq]~(\tfp \, r)
\]
and similarly for $\tfq \cdot_r \tfp$. We define $u f$ in
$\TTF^{\text{pre}}(J)$ for $f : J \to I$ by induction on $u$ just like
for $\Sp^1_{\text{pre}}$, the interesting case is:
\begin{align*}
 (\tfsurf \, r \, s) f =&
  \begin{cases}
    \tfsurf \, (r f) \, (s f) & \text{if } r f \neq 0,1 \text{ and } sf \neq 0,1 \\
    \tfp \cdot_{sf} \tfq & \text{if } rf = 0 \text{ and } sf \neq 0,1 \\
    \tfq \cdot_{sf} \tfp & \text{if } rf = 1 \text{ and } sf \neq 0,1 \\
    \tfb      & \text{otherwise}\\
 \end{cases}  \\
 (\hcomp~[\phi\mapsto u]~u_0)f =&
  \begin{cases}
   u_{f,1} & \text{if } \phi f = 1 \\
   \hcomp~[\phi f\mapsto uf^+]~(u_0 f) & \text{otherwise}
 \end{cases}
\end{align*}
where $uf^+$ is the family $(uf^+)_{g,r} = u_{fg,r}$ for $g : K \to J$.

We then define the subset $\TTF(I) \subseteq \TTF^{\text{pre}}(I)$ by
taking the elements $\tfb$, $\tfp \, r$, $\tfq\,r$, $\tfsurf\,r\,s$
and $\hcomp~[\phi\mapsto u]~u_0$ such that $u_0 \in \TTF(I)$, $u_{f,r} \in
\TTF(J)$ for $f:J\to I$ satisfying $u_0 f = u_{f,0}$ for $f:J\to
I$ and $u_{f,r} g = u_{fg,rg}$ for $f : J \to I$ and $g : K \to J$.  This
defines a cubical set $\TTF$, such that $\TTF(I)$ is a subset of
$\TTF^{\text{pre}}(I)$ for each $I$.

\paragraph{Suspension}
Given presheaf $\Gamma$ and $A$ a dependent presheaf over $\Gamma$
(which is a presheaf on the category of elements of $\Gamma$) we
explain how to build the suspension of $A$, written $\susp{A}$, which
is an initial $\susp{A}$-algebra.  Just like for the parameter-free
higher inductive types we first define a family of sets $X (I,\rho)$,
for $\rho \in \Gamma(I)$ which is an upper approximation of the
suspension. An element of $X(I,\rho)$ is of the form:
\begin{itemize}
  \item $\north$, $\south$, or
  \item $\merid~a~r$ with $a \in A (I,\rho)$ and $r \neq 0,1$ in
    $\II(I)$, or
  \item $\hcomp~[\phi\mapsto u]~u_0$ with $\phi \neq 1$ in $\FF(I)$
    and $u_0$ in $X(I,\rho)$ and $u$ a family of elements $u_{f,r}$ in $X(J,
    \rho f)$ for $f \co J \to I$ such that $\phi f = 1$.
\end{itemize}

In this way an element of $X(I,\rho)$ can be seen as a well-founded
tree. We now define the tentative restriction maps $X(I,\rho) \to
X(J,\rho f)$, $u \mapsto uf$ for $f \co J \to I$ by induction on $u$:
\begin{align*}
  \north f =& ~ \north \\
  \south f =& ~ \south \\
  (\merid\,a\,r) f =&
  \begin{cases}
   \north & \text{if } rf = 0 \\
   \south & \text{if } rf = 1 \\
   \merid\,(af)\,(rf)      & \text{otherwise}
 \end{cases} \\
  (\hcomp~[\phi\mapsto u]~u_0)f =&
  \begin{cases}
   u_{f,1} & \text{if } \phi f = 1 \\
   \hcomp~[\phi f\mapsto uf^+]~(u_0 f) & \text{otherwise}
 \end{cases}
\end{align*}
where $uf^+$ is the family $(uf^+)_{g,r} = u_{fg,rg}$ for $g : K \to J^+$.

We define $(\susp{A}) (I,\rho)$ as the subset of $X(I, \rho)$ of
elements $\north$, $\south$ or $\merid~a~r$ with $a \in A \rho$ and
$\hcomp~[\phi\mapsto u]~u_0$ with $u_0$ in $(\susp{A})\rho$ and
$u_{f,0} = u_0 f$ for $f:J\to I$ and each $u_{f,r}$ in
$(\susp{A})(\rho f)$ for $f:J\to I$ and $u_{f,r}g = u_{fg,rg}$ for $g : K
\to J$ and $f:J\to I$ and $r$ in $\II(J)$.

This defines the initial $\susp{A}$-algebra relative to a context
$\Gamma$.  Since this operation commutes with substitution $\Delta \to
\Gamma$, it is an external description of the operation which takes an
arbitrary type $A$ and produces the free $\susp{A}$-algebra.

\paragraph{Pushouts}
Given $D = A, B, C, u:C\to A, v:C\to B$ a diagram over $\Gamma$ we
explain how to define $\Push$, initial $D$-algebra over $\Gamma$.  We
first define a family of sets $X(I,\rho)$, for $\rho \in \Gamma(I)$
which is an upper approximation of the pushout. An element of
$X(I,\rho)$ is of the form:
\begin{itemize}
  \item $\inl\,a$ for $a \in A (I,\rho)$, or
  \item $\inr\,b$ for $b \in B (I,\rho)$, or
  \item $\push~c~r$ with $c \in C(I,\rho)$ and $r \in \II(I)$ such
    that $r \neq 0,1$, or
  \item $\hcomp~[\phi\mapsto u]~u_0$ with $\phi \neq 1$ in $\FF(I)$
    and $u_0$ in $X(I,\rho)$ and $u$ a family of elements $u_f$ in
    $X(J,\rho f)$ for $f \co J \to I$ such that $\phi  f = 1$.
\end{itemize}

The maps $X(I,\rho) \to X(J,\rho f)$ for $f \co J \to I$ are defined
by induction:
\begin{align*}
  (\inl\,a) f =& ~ \inl\,(a f) \\
  (\inr\,a) f =& ~ \inr\,(b f) \\
  (\push\,c\,r) f =&
  \begin{cases}
   \inl\,(\app(u,c f)) & \text{if } rf = 0 \\
   \inr\,(\app(v,c f)) & \text{if } rf = 1 \\
   \push\,(cf)\,(rf)      & \text{otherwise}
 \end{cases} \\
  (\hcomp~[\phi\mapsto u]~u_0)f =&
  \begin{cases}
   u_{f,1} & \text{if } \phi f = 1 \\
   \hcomp~[\phi f\mapsto uf^+]~(u_0 f) & \text{otherwise}
 \end{cases}
\end{align*}
where $uf^+$ is the family $(uf^+)_{g,r} = u_{fg,r}$ for $g : K \to J$.

We define $(\Push)(I,\rho)$ for $\rho \in \Gamma(I)$ as the subset of
$X(I,\rho)$ with elements
\begin{itemize}
  \item $\inl~a$ with $a \in A(I,\rho)$, or
  \item $\inr~b$ with $b \in B(I,\rho)$, or
  \item $\push~c~r$ with $c \in C(I,\rho)$ and $r \in \II(I)$ such
    that $r \neq 0,1$, or
  \item $\hcomp~[\phi\mapsto u]~u_0$ with $u_0$ in $(\Push)(I,\rho)$
    and $u_{f,0} = u_0 f$ if $f\co J\to I$ and each $u_f$ in
    $(\Push)(J,\rho\sigma f)$ for $f\co J\to I$ and $u_{f,r}g = u_{fg,rg}$
    for $g \co K \to J$ and $f \co J\to I$ and $r$ in $\II(I)$.
\end{itemize}

\end{document}